\documentclass[a4paper]{amsart}
\usepackage[latin1]{inputenc}
\usepackage{amsmath}
\usepackage{amsfonts}
\usepackage{amssymb}
\usepackage{graphicx}
\usepackage{color}
\numberwithin{equation}{section}
\newtheorem{theorem}{Theorem}[section]
\newtheorem{lemma}[theorem]{Lemma}
\newtheorem{proposition}[theorem]{Proposition}
\newtheorem{corollary}[theorem]{Corollary}
\theoremstyle{definition}
\newtheorem{definition}[theorem]{Definition}
\newtheorem{example}[theorem]{Example}

\theoremstyle{remark}

\newcommand{\lbiprod}{{>\!\!\!\triangleleft\kern-.33em\cdot}}
\newcommand{\rbiprod}{{\cdot\kern-.33em\triangleright\!\!\!<}}

\makeatletter
\newcommand{\mathleft}{\@fleqntrue\@mathmargin0pt}
\newcommand{\mathcenter}{\@fleqnfalse}
\makeatother

\newcommand{\C}{{\Bbb C}}
\newcommand{\R}{{\Bbb R}}
\newcommand{\Z}{{\Bbb Z}}

\newcommand{\cg}{\mathfrak{g}}

\newcommand{\eps}{\epsilon}

\newcommand{\tens}{\otimes}
\newcommand{\id}{{\rm id}}
\newcommand{\extd}{{\rm d}}

\newcommand{\del}{{\partial}}

\allowdisplaybreaks

\title{Time-slicing quantum spacetimes}


\author{Shahn Majid}
\address{School of Mathematical Sciences\\ Queen Mary University of London \\ Mile End Rd, London E1 4NS }
\thanks{Work supported by a Leverhulme Trust project grant RPG-2024-177}
\email{s.majid@qmul.ac.uk}
\begin{document}

\begin{abstract} For quantum field theory on curved spacetimes, a critical role is played by their foliation into  spacelike time-slices at each value $t$ of a coordinate time, with corresponding metric in ADM form. We provide a general construction for the spacetime quantum Levi-Civita connection when each spatial slice is replaced by a quantum Riemannian geometry. This is then fully solved for a class of spatial algebras including fuzzy spheres and  for any time-dependent spatial quantum metric, shift 1-form and lapse function. The result takes a particularly simple form if the spatial metric evolves in time according to a first order ODE which, in the case of a fuzzy sphere, requires the spatial metric to rotate in time according to the value at each $t$ of the shift vector. As an application, our results  provide in principle fuzzy versions of most (pseudo)-Riemannian manifolds. We also fully solve the case of rotationally invariant spacetimes with angular directions replaced by a discrete circle, including a new $\Z_n$-FLRW model.\end{abstract}
\maketitle 

\section{Introduction}

In this work, we take a first step towards canonical quantum gravity in the ADM formalism\cite{ADM} adapted to quantum spacetimes. The thinking is that if, as widely considered plausible, quantum gravity effects render spacetime noncommutative, we should build quantum field theories on quantum spacetimes. This should include quantum gravity itself, with the ultimate hope of producing the pre-supposed quantum gravity effects as a self-consistency condition. While many quantum spacetime models are known, starting with\cite{MaRue, DFR,Hoo}, and also include lattices and graphs as examples (where the algebra of functions is commutative but functions do not commute with 1-forms in quantum Riemannian geometry (QRG)\cite{BegMa}), quantum field theory remains poorly understood. It is not hard to write down quantum-geometric Lagrangians and hence propose a functional integral quantisation, eg for quantum gravity on a square\cite{Ma:squ,BliMa}, but the full physical picture also needs a Hamiltonian quantisation, and this needs a foliation into time-slices that works in quantum  geometry. 

In classical geometry, at least locally, this is not considered an issue, but rather more of a gauge choice. Thus, in the ADM formalism, one assumes the metric on the spacetime $M$ can be put into the form\cite{ADM}
\begin{equation}\label{ADM} \cg=g_{ij}\extd x^i\tens \extd x^j + g_{ij}N^j(\extd x^i\tens \extd t+\extd t\tens \extd x^i)- (N^2-g_{ij}N^iN^j)\extd t\tens\extd t,\end{equation}
where we follow algebraic conventions in writing the tensor products (taken over $C^\infty(M)$). Here $x^i$ are coordinates in the time-slice at time $t$ 
and the corresponding points in $M$ are denoted $x^\mu(\vec x,t)$. The coefficients are functions of $x^i$ that also depend on $t$, with $g_{ij}$ the spatial metric on the slice, $N$ governing the lapse of proper time between nearby slices and $N^j$ governing the corresponding shift in spatial coordinates. In this work, we show how to replace the $x^i$ by a noncommutative spatial coordinate algebra $A$ with differential structure and a free choice of spatial quantum metric $\cg(t)$, shift 1-form $\xi(t)$ (corresponding to a shift vector with components $N^i$ via the spatial metric) and lapse function $n(t)$ in the role of $N^2(t)$ (we wont actually take its square-root as this might be an issue in the algebra $A$). There are different approaches to noncommutative geometry and we use the `quantum Riemannian geometry' (QRG) formalism as in \cite{BegMa} that grew out of experience with concrete models and quantum group symmetries, so is well suited for direct calculations. Then $\cg\in \Omega^1\tens_A\Omega^1$ and we assume we are given a quantum Levi-Civita connection (QLC)  $\hat\nabla:\Omega^1\to \Omega^1\tens_A\Omega^1$, i.e. torsion free and metric compatible for every $\cg$. The $\tens_A$ means we can move elements of $A$ freely across the tensor product. The equations for a QLC in the noncommutative case are nonlinear (quadratic) hence this may not exist or it may not be unique and the problem we solve is to find a spacetime  $\tilde\nabla$ which is a QLC for a spacetime quantum metric of the form
\begin{equation}\label{gintro} \tilde\cg=\cg+ \xi\tens \extd t + \extd t\tens\xi - (n-(\xi,\xi))\extd t\tens\extd t,\end{equation}
where $(\ ,\ ):\Omega^1\tens_A\Omega^1\to A$ is the inverse metric to $\cg$ and $\cg,\xi,n$ are freely time-dependent for a commutative central variable $t$ which we have adjoined to $A$. Details of the QRG formalism here are briefly recalled in Section~\ref{secpre}. Our main result, in Section~\ref{secmain} it to analyse what is needed for this and how to obtain $\tilde\nabla$ in principle, with an explicit solution in Section~\ref{seccen} in the case where $\Omega^1$ can a central basis and some associated assumptions. In Section~\ref{secdisc} we also have some partial results for the complementary case where the spatial QRG has a discrete aspect (and hence no central basis), including a complete solution for the angular geometry given by $\Z_n$ and a rotationally invariant spacetime quantum metric. 

We have been motivated by work \cite{Will} on a `discrete Wheeler-De-Witt equation' with the spatial geometry discrete and described by some form of Regge calculus. This does not fit into the QRG formalism but on the other hand the authors were able to go further to construct conjugate $\pi_{ij}$ etc. in relation to the Einstein-Hilbert action and variational calculus. The QRG version of extrinsic curvature to take this further is not clear and nor is noncommutative variational calculus (although with recent progress\cite{MaSim}), but in the meantime our results are still useful as a way to construct foliated quantum spacetimes in the first place. Moreover, we can use the same construction but for a spatial variable in the role of $t$ to build the spatial $A$ in the first place, for example starting with an angular quantum geometry and adjoining $r$ (in the role of $t$ above) and then iterating to adjoin $t$. By such iteration, we can similarly insert a sufficiently nice initial noncommutative algebra, including fuzzy spheres\cite{Mad,BegMa}, into most (pseudo)-Riemannian spacetimes as quantum versions. This provides a rich supply of examples of spacetime QRGs (and in principle recovers the fuzzy sphere Kaluza-Klein theory in\cite{LiuMa} in a different context). In particular, we will recover known fuzzy black-hole and fuzzy FLRW spacetimes and the $\Z_n$-black-hole in \cite{ArgMa}, but now understood more systematically and now completed with a $\Z_n$-FLRW quantum spacetime in Example~\ref{FLRW}. Some concluding remarks about directions for further work are in Section~\ref{seccon}.  

\section{Elements of quantum Riemannian geometry}\label{secpre}

The basic ingredients for a quantum Riemannian geometry (QRG) are $A$ a unital algebra in the role of coordinate algebra or `functions' in the geometry, but allowed to be noncommutative. A differential calculus is a graded algebra $\Omega=\oplus_n\Omega^n$ where $\Omega^0=A$, and $\extd: \Omega^n\to \Omega^{n+1}$ is the exterior derivative obeying a graded-Leibniz rule. In addition, we ask that $\Omega$ is generated by $A,\extd A$ so every form can be generated by functions and differentials. Included here is the Leibniz rule $\extd(ab)=a\extd b+(\extd a)b$ for all $a,b\in A$. The product among forms of positive degree is denoted $\wedge$ while the product by elements of $A$ in degree 0 is denoted as here by omission (and makes each $\Omega^n$ into an $A$-bimodule). If there is a basis $\{s^i\}$ of $\Omega^1$ (say a left one) then $\extd a=(\del_i a) s^i$ for all $a\in A$ (sum understood) defines partial derivatives $\del_i:A\to A$ with respect to the basis. If the $s^i$ commute with functions (a central basis) then the $\del_i$ are derivations.  In physical applications, we want $A$ to be a $*$-algebra (for an actual space, it would be complex-valued functions with $*$ given by complex conjugation) so $*^2=\id$ and $(ab)^*=b^*a^*$. We want this to extend to $\Omega$ in such a way that $\extd$ commutes with $*$ and $*$ is graded-antimultiplicative. We will ask for our bases, if they exist, to be self-adjoint so $(s^i)^*=s^i$.  

A metric is $\cg\in \Omega^1\tens_A\Omega^1$ admitting an inverse in the sense of a map $(\ ,\ ):\Omega^1\tens_A\Omega^1\to A$ which is well-defined for the tensor product over $A$ (where elements of $A$ can be moved through $\tens_A$) and is a bimodule map (respects the product by elements of $A$ from either side). That  it is inverse is expressed as 
\[ g^1(g^2,\omega)=\omega=(\omega,g^1)g^2\]
for all $\omega\in\Omega^1$, where $\cg=g^1\tens g^2$ is a notation for the two parts of $\cg$ (a sum of such terms is to be understood). Remarkably, one can deduce that
\begin{equation}\label{gcen} a\cg =\cg a\end{equation}
for all $a\in A$ even though we have not assumed that $\Omega$ is graded-commutative in any way. In addition, we ask for $\cg$ to be symmetric in some fashion. One approach to this is $\wedge(g)=0$ but in practice, e.g., on graphs, other variants could be more appropriate. Otherwise, $\cg$ is a generalised metric (it could even be antisymmetric). The `reality' condition is 
\[ {\rm flip}(*\tens *)(\cg)=\cg\]
which in the case where there is a self-adjoint central basis amounts to the matrix of coefficients defined by $\cg=g_{ij}s^i\tens s^j$ being hermitian with constant entries due to (\ref{gcen}). This will be the case for a fuzzy sphere.  Obviously, it is possible to weaken the axioms and lose (\ref{gcen}) but the requirements that lead to (\ref{gcen}) are natural and hence qualify as the first thing to look at even if other options are possible.

Next, we need a quantum Levi-Civita connection i.e. metric compatible and torsion free. The latter in this `algebraic' way of doing Riemannian geometry just amounts to vanishing of
\[ T_\nabla:=\wedge\circ\nabla-\extd: \Omega^1\to \Omega^2, \]
where a (left) connection $\nabla$ is defined as
\[ \nabla:\Omega^1\to \Omega^1\tens_A\Omega^1,\quad \nabla(a\omega)=\extd a\tens\omega+ a\nabla\omega\]
for all $a\in A, \omega\in \Omega^1$. The left most copy of $\Omega^1$ is where we could potentially apply a quantum vector field $X:\Omega^1\to A$ (defined as a right $A$-module map, i.e. respecting the action of $A$ from the right), which in the classical limit would correspond to a covariant derivative $\nabla_X$.  A special `bimodule' type of left connection is one that also admits a right Leibniz rule
\[ \nabla(\omega a)=(\nabla\omega)a+ \sigma(\omega\tens\extd a),\quad \sigma:\Omega^1\tens_A\Omega^1\to \Omega^1\tens_A\Omega^1\]
for some bimodule map $\sigma$ (this `generalised braiding' is not additional data, if it exists its uniquely determined by $\nabla$). In this case, 
$\nabla$ extends to tensor products so that we can make sense of 
\[ \nabla\cg:= \nabla g^1\tens g^2 + (\sigma(g^1\tens\ )\tens\id)\nabla g^2.\]
The role of $\sigma$ is to move the first leg of $\nabla g^2$ to the far left so that one could potentially evaluate a vector field against the left-most factor in a consistent manner. In classical geometry, this $\sigma$ would be taken for granted as a flip of tensor factors and connect to such a vector field as expected. We are now in position to define a QLC as a bimodule connection that obeys  $T_\nabla=\nabla\cg=0$. These are effectively linear plus quadratic constraints on the coefficients of $\nabla$ (because $\sigma$ depends linearly on $\nabla$) with the result that a solution may not exist or may exist but may not be unique. This non-linear aspect is missing classically when $\sigma$ is just the flip. We also want compatibility with $*$ in the sense 
\[ \nabla (\omega^*)=\sigma\circ {\rm flip}(*\tens *)\circ\nabla\omega\]
for all $\omega\in \Omega^1$, again generically not linear in $\nabla$. If we expand in a basis according to $\nabla s^i= - \Gamma^i{}_{jk} s^j\tens s^k$ and {\em if} the bases are self-adjoint, central and $\sigma$ on the basis is a flip, then this amounts to $\Gamma^i{}_{jk}$ real. This will be the case for the fuzzy sphere, where there is (under some mild assumptions) a unique QLC determined by the metric coefficients $g_{ij}$, see \cite{LirMa,BegMa}. 

It is worth noting that in any QRG, torsion freeness and the connection preserving $*$ imply respectively that\cite[Chap. 8]{BegMa}
\[ \wedge\circ(\id+\sigma)=0,\quad {\rm flip}(*\tens *)\circ \sigma=\sigma^{-1}\circ{\rm flip}(*\tens *).\]
We restrict attention to connections with invertible $\sigma$. The formalism goes on to the construction of the Riemann curvature $R_\nabla$, which is canonical, and Ricci curvature which is a working definition in the sense that it needs additional data of a `lift' bimodule map $i:\Omega^2\to \Omega^1\tens_A\Omega^1$ such that following this with $\wedge$ is the identity map. This allows to separate and hence contract the indices of the Riemann tensor but lacks a more conceptual origin (which likely would need noncommutative variational calculus). More details and examples are in \cite{BegMa}. 

\section{Extending spatial quantum geometries by central time}\label{secmain}

Here we analyse what happens if we are given a spatial QRG on an algebra $A$ with exterior algebra $(\Omega,\extd)$, metric $\cg$ and quantum Levi-Civita connection $\hat\nabla$, and want to view it as a slice at time $t=0$ of a quantum spacetime and construct the QLC $\tilde\nabla$ for that. We won't assume that  its restriction to $A$ is $\hat\nabla$ but could be some other torsion free connection $(\nabla,\sigma)$ (related to $\hat\nabla$ by addition of a bimodule map). In general this is a hard question and to get off the ground we have to first decide what properties $t$ and $\extd t$ should have. In this paper, we take the simplest case that $t$ commutes with everything and $\extd t$ graded-commutes with everything. It means that $t$ is a classical time variable and parametrises each time slice. We take the spacetime coordinate algebra $\tilde A:=A\tens C^\infty(\R)$ (or better, extend this to $C^\infty(\R, A)$ for a suitable notion of smoothness which we will not make precise). For differential forms, we similarly take $\tilde\Omega:=\Omega\underline{\tens} \Omega(\R)$ where the two factors graded-commute. Here 
\[\tilde\Omega^n=\Omega^n\tens C^\infty(\R) \oplus \Omega^{n-1}\tens C^\infty(\R)\extd t.\]

In this case, the most general quantum metric has the form
\begin{equation}\label{tg} \tilde\cg=\cg + \xi\tens\extd t+ \extd t\tens\eta+ N\extd t\tens\extd t\end{equation}
for 1-forms $\xi,\eta\in \Omega^1$ and a function $N\in A$ (in the role of $(\xi,\xi)-n$ in (\ref{gintro}), by a small abuse of notation), except that we allow these and the spatial quantum metric $\cg$ to now depend on $t$.  We will later set $\xi=\eta$ as some form of quantum symmetry, but for now we keep our options open. At all times,  we will need to maintain
\[ [a,\cg]=[a,\xi]=[a,\eta]=[a,N]=0\]
for all $a\in A$ so that the metric can remain bimodule. Next, we take the QLC again in the most general form
\begin{align} \tilde\nabla\omega&=\nabla\omega+ \alpha(\omega)\tens\extd t+ \extd t\tens \beta(\omega) + \gamma(\omega)\extd t\tens\extd t, \label{tcon1}\\
\tilde\nabla\extd t&=\Xi +   C\tens\extd t+\extd t\tens D+ E\extd t\tens\extd t \label{tcon2}\end{align}
for all $\omega\in \Omega^1$ (i.e., on $A$). Here 
\[ \alpha,\beta:\Omega^1\to \Omega^1,\quad \gamma:\Omega^1\to A,\quad \Xi\in \Omega^1\tens_A\Omega^1,\quad C,D\in \Omega^1,\quad E\in A\]
except that everything is also allowed to depend on $t$. The conditions we impose at this stage apply for every fixed $t$ so we keep this in the background for now. Indeed, because $T_{\tilde\nabla}$ and $\tilde\sigma$ are extended as left-module and bimodule maps respectively, it suffices to test the conditions on time-independent elements as any time-dependence in coefficients can be factored out term by term. 

\begin{lemma}\label{lem:tor} For $\tilde \nabla$ to have zero torsion, we need $T_\nabla=0$, $C=D, \alpha=\beta$ and $\wedge(\Xi)=0$. \end{lemma}
\begin{proof} This follows from 
\[ \wedge\tilde\nabla\omega-\extd\omega=T_\nabla(\omega)+ (\alpha(\omega)-\beta(\omega))\wedge\extd t,\quad  \wedge\tilde\nabla\extd t-\extd^2t= \wedge(\Xi)+ (C-D)\wedge\extd t\]
since $\extd^2 t=0$ and $\extd t$ anticommutes with $\Omega^1$. We then need the 
terms in different spaces to vanish separately.  \end{proof}

\begin{proposition}\label{prop:sig} A natural and sufficient condition for $\sigma$ to exist as a bimodule map is that $\alpha,\beta,\gamma$ are bimodule maps and $\Xi,C,D,E$ are central. In this case $\tilde\sigma$ is the flip when either argument is $\extd t$ and $\sigma$ otherwise. 
\end{proposition}
\begin{proof} Subtracting the two Leibniz rules for $\tilde\nabla$, and assuming that $\tilde\sigma$ exists, one arrives as
\begin{align*} \tilde\sigma(\extd t\tens \extd a)&=\extd a\tens\extd t+ [a, \Xi+C\tens\extd t+ \extd t\tens D+ E\extd t\tens\extd t]\\
\tilde\sigma(\omega\tens \extd t)&= \extd t\tens\omega \\
\tilde\sigma(\extd t\tens \extd t)&= \extd t\tens\extd t \\
\tilde\sigma(\omega\tens \extd a)&= \sigma(\omega\tens\extd a)+   [a, \alpha(\omega)\tens\extd t+ \extd t\tens \beta(\omega)+ \gamma(\omega)\extd t\tens\extd t]\\
&\quad + \alpha([\omega,a])\tens\extd t+ \extd t\tens\beta([\omega,a])+ \gamma([\omega,a])\extd t\tens\extd t.
\end{align*}
For example, the last case comes from expanding $\tilde\nabla(\omega a)-\tilde\nabla(a\omega)$ using the two Leibniz rules and rearranging as
\begin{align*}\tilde\sigma(\omega\tens \extd a)&=[a,\tilde\nabla\omega]+\tilde\nabla([\omega, a])+\extd a\tens \omega.
\end{align*} 
We then expand this out in terms of the definition of $\tilde\nabla$ and also recognise the spatial $\sigma$ among the terms. The other cases are similar but easier so we omit them. Now, for this to be well-defined we need that the terms involving $[a, ]$ depend only on $\extd a$, which depends on the relations in $\Omega^1$. But one simple solution is the one stated, given that $\extd t$ commutes with $A$. Then terms involving $[a, ]$ all cancel out. One check that there are then no further conditions from $\tilde\sigma$ being a well-defined bimodule map. 
\end{proof}

We will now proceed under the requirements stated in the lemma and proposition, to have a bimodule torsion-free connection. As for the metric, we write $\Xi=\Xi^1\tens\Xi^2\in \Omega^1\tens_A\Omega^1$ with a sum of terms understood. We also use dot for $\extd\over\extd t$. 

\begin{theorem}\label{thm:nabg} A torsion free bimodule connection $\tilde\nabla$ as in Lemma~\ref{lem:tor} and Proposition~\ref{prop:sig} is metric compatible with $\tilde g$ iff
\begin{align} \nabla\cg+\Xi\tens\eta+ \sigma(\xi\tens\Xi^1)\tens \Xi^2=0, \label{g1}\\
 \sigma(g^1\tens\alpha(g^2)+\xi\tens C)+ \nabla\xi+ N\Xi=0, \label{g2}\\
\alpha(g^1)\tens g^2+ C\tens\eta + \nabla\eta + N\Xi=0, \label{g3}\\
 \dot\cg+ g^1\tens\alpha(g^2)+ \alpha(g^1)\tens g^2+C\tens\eta+\xi\tens C=0, \label{g4}\\
\extd N+ 2NC+ \alpha(\xi+\eta) =0, \label{g5}\\
 \dot\xi+ E\xi + \alpha(\xi) + N C+ g^1\gamma(g^2)=0, \label{g6}\\
 \dot\eta+ E\eta+ \alpha(\eta)+ NC+ \gamma(g^1)g^2=0, \label{g7}\\
\dot N+ 2E N + \gamma(\xi+\eta)=0, \label{g8},
\end{align}
where $\extd N$ is the spatial exterior derivative. 
\end{theorem}
\begin{proof} First, we need to look more carefully at how connections $\tilde\nabla$ of the form (\ref{tcon1})-(\ref{tcon2}) act on time-dependent 1-forms. If $\tilde\omega=a(t)\omega$ where $\omega$ is time-independent then 
\begin{align*}\tilde\nabla\tilde\omega&=\extd a(t)\tens\omega+ a(t)\tilde\nabla\omega=\dot a(t)\, \extd t\tens\omega+ a(t) \tilde\nabla\omega\\
&=\extd t\tens\dot{\tilde\omega}+ a(t)\nabla\omega+ a(t)\alpha(\omega)\extd t+ a(t)\extd t\tens \beta(\omega)+a(t)\gamma(\omega)\extd t\tens\extd t.\end{align*}
What this means is that when a 1-form $\tilde\omega$ is time-dependent, we can use the same formula (\ref{tcon1}) ignoring the time-dependence by the convention that $\alpha,\beta,\gamma$ commute with multiplication by $a(t)$ (so in the above, we can take it inside and recombine to $\tilde\omega$) but we pick up an extra term $\extd t\tens \dot{\tilde\omega}$. Similarly, if we have $\tilde N=a(t) N$ say, then $\tilde\nabla(\tilde N \extd t)=\extd (a(t)N)\tens\extd t+ \tilde N\tilde\nabla\extd t= \dot{\tilde N}\extd t\tens\extd t+ a(t)\tilde\nabla(N\extd t)$ so we get an extra term from the time dependence compared to what we would have otherwise. With this in mind we let $\omega,N$ etc in what follows be time-dependent but proceed as if they are not, with the extra $\dot\omega,\dot N$ etc terms added where there was a covariant derivative applied.

Next, proceeding for brevity in the torsion-free case, we insert (\ref{tg}), (\ref{tcon1})-(\ref{tcon2}) into $\tilde\nabla\tilde g$. We use $\tens_S$ for a tensor product  plus its flip as a short-hand and ignore time dependence for now. Then we obtain
\begin{align*} \tilde\nabla\tilde g&=\nabla g^1\tens g^2+ \alpha(g^1)\tens_S\extd t \tens g^2+ \gamma(g^1)\extd t\tens\extd t\tens g^2+ \extd N\tens\extd t\tens\extd t\\
&\quad + N(\Xi+ C\tens_S\extd t+ E\extd t\tens\extd t)\tens\extd t+ \nabla\xi\tens \extd t+ \alpha(\xi)\tens_S\extd t+ \gamma(\xi)\extd t\tens\extd t\tens\extd t\\
&\quad + (\Xi+C\tens_S\extd t+ E\extd\tens\extd t)\tens\eta\\
&\quad +(\tilde\sigma\tens\id)\Big(g^1\tens\nabla g^2+ g^1\tens\alpha(g^2)\tens_S\extd t+ g^1\tens\gamma(g^2)\extd t\tens\extd t\\
&\quad + (N\extd t+\xi)\tens(\Xi+ C\tens_S\extd t+E\extd t\tens\extd t)+\extd t\tens(\nabla\eta+\alpha(\eta)\tens_S\extd t+ \gamma(\eta)\extd t\tens\extd t)\Big).  \end{align*}
Next, we use that $\tilde\sigma$ in Proposition~\ref{prop:sig} is the flip when one argument is $\extd t$, and $\sigma$ otherwise. With this mind, we equate to zero each of the different components in $\Omega^1\tens_A\Omega^1\tens_A\Omega^1, \Omega^1\tens_A\Omega^1\tens_A\extd t, \Omega^1\tens_A\extd t\tens_A\Omega^1, \extd t\tens_A\Omega^1\tens_A\Omega^1, \Omega^1\tens_A\extd t\tens_A\extd t, \extd t\tens_A\Omega^1\tens_A\Omega^1, \extd t\tens_A\extd t\tens_A\Omega^1+ \extd t\tens_A \extd t\tens_A\extd t$ in that order. Note that the first term above and the first term inside $\tilde\sigma$ gives $\nabla\cg$ for the spatial metric. The result is the 8 equations listed in the time-independent case. We then review where $\tilde\nabla$ was applied to a potentially time dependent function or 1-form and add terms to the 4th, 6th, 7th and 8th equations from this. Here $\extd t\tens \dot{g^1}\tens g^2+ \tilde\sigma(g^1\tens \extd t\tens \dot{g^2})=\extd t\tens \dot g$.  \end{proof}

Note that it is easy enough to `polarise' these formula if we want to drop the torsion freeness (by replacing $C\tens_S\extd t$ and $\alpha(\xi)\tens_S\extd t$ etc, by the general expressions so we have both $\alpha,\beta$ and $C,D$.) It remains to see when $\tilde g$ and $\tilde \nabla$ obey the `reality' conditions for a $*$-structure on the spatial $(A, \Omega,\extd)$. We extend this to $\tilde A$ by $t=t^*$ as this variable is treated classically.

\begin{corollary}\label{cor:star} (1) A metric $\tilde\cg$ of the form (\ref{tg}) obeys the `reality' condition iff $\cg$ does and $N^*=N$, $\xi^*=\eta$.

(2) A bimodule connection $\tilde\nabla$ as in Proposition~\ref{prop:sig}  obeys the `reality' ($*$-preserving) condition iff $\nabla$ does, $\alpha,\beta,\gamma$ commute with $*$ and
\[ C^*=C,\quad D^*=D,\quad E^*=E,\quad \sigma\circ{\rm flip}(*\tens *)(\Xi)=\Xi.\] 
\end{corollary}
\begin{proof}  We make $\tilde\Omega$ into a $*$-calculus in the obvious way with $t$ and hence $\extd t$ self-adjoint. Then the conditions for $\tilde g$ of the form (\ref{tg}) apply at each $t$ and using this form, ${\rm flip}(*\tens *)(\tilde g)=\tilde g$ immediately gives the result on looking at the different components. The $\Omega^1\tens_A\Omega^1$ component is just ${\rm flip}(*\tens *)(g)=g$ for the spatial metric. Next, $*$ commutes with time derivatives hence, while time-dependence in the coefficients will give some extra terms from $\tilde\nabla$, these will cancel and we can test the requirement on time-independent elements. We then use the form (\ref{tcon1})-(\ref{tcon2}) to expand out $\tilde\sigma\circ{\rm flip}(*\tens *)\circ\tilde\nabla$ applied to $\omega, \extd t$ and compare with $\tilde\nabla$ applied to $\omega^*,\extd t$. The $\Omega^1\tens_A\Omega^1$ components need $\sigma\circ{\rm flip}(*\tens *)\circ\nabla\omega=\nabla(\omega^*)$ and $\sigma\circ{\rm flip}(*\tens *)(\Xi)=\Xi$, respectively. We used that $\tilde\sigma$ is just the flip if an argument is $\extd t$ and $\sigma$ otherwise as in Proposition~\ref{prop:sig}. \end{proof}

Note that the equations in Theorem~\ref{thm:nabg} are either invariant or fall into pairs related by $*$ (namely (\ref{g2})/(\ref{g3}) and (\ref{g6})/(\ref{g7})). They also substantially overlap directly if we impost quantum symmetry in the form $\wedge(\tilde{g})=0$, assuming $\wedge(g)=0$. This amounts to
\begin{equation}\label{tgsim} \xi=\eta \end{equation}
or in the $*$-algebra setting it amounts to $\xi=\xi^*$, i.e. with real coefficients with respect to a self-adjoint basis. We henceforth focus on this case. But note as a warning that the conjugate pairs are only equivalent if we know that $\alpha,\gamma$ commute with $*$ and in fact imposing both of them is equivalent to this. So we have to solve both halves even when $\xi=\eta=\xi^*$. 

We are also free to substitute all the terms in (\ref{g4}) after $\dot\cg$ using (\ref{g2})-(\ref{g3}), assuming that $\sigma$ is invertible (which we do). In the self-adjoint case (\ref{g4}) is then equivalent to
\begin{equation}\label{g9} \dot\cg =(\id+\sigma^{-1})(N\Xi+ \nabla\xi).\end{equation}

We now explain how these equations can be solved,  subject to some details which depend on the specific QRG $A$. First, note that  if $\nabla$ is torsion free then ${\rm Image}(\id+\sigma)\subseteq\ker\wedge$ by a result in \cite[Eq. (3.30)]{BegMa} and hence the same is true for $(\id+\sigma^{-1})$. If this is an equality then we might expect the map restricted to $\ker\wedge$ to be invertible. If so then (\ref{g9}) has a unique solution for $\Xi$ such that $\wedge(\Xi)=0$. Here $N$ is assumed invertible and $(\id+\sigma^{-1})\nabla\xi$ is moved to the left hand side so that if $\cg$ is quantum symmetric in the wedge sense then the new left hand side also lives in $\ker\wedge$. We proceed in the self-adjoint case where $\eta=\xi^*=\xi$, $N^*=N$, with $N$ invertible, and fix a QLC $\hat\nabla$ for the spatial metric with braiding $\sigma$. Finally, given that $N,\xi$ are central, it is also natural to assume that $\extd N$ is central and
\begin{equation}\label{cenxi} (\xi,\ )=(\ ,\xi),\quad \sigma(\xi\tens\ )=(\ )\tens\xi.\end{equation}
The first of these ensures that if $\omega$ self-adjoint then so is $(\xi,\omega)$, while the second ensures that $\hat\nabla\xi$ is central. 

\begin{proposition}\label{prop:gen} In our setting, suppose that $\xi$ obeys (\ref{cenxi}) and $\extd N$ is central,  $n$ is invertible and that a central $\Xi$ exists solving 
\[ \dot\cg=(\id+\sigma^{-1})(\hat\nabla\xi-n\Xi),\quad \wedge(\Xi)=0\]
and obeying the `reality'  condition  in Corollary~\ref{cor:star}(2). Then there is a unique solution to (\ref{g1}),(\ref{g3}),(\ref{g5}),(\ref{g6}),(\ref{g8}), which makes $\tilde\nabla$ a QLC iff 
\[\sigma(g^1\tens \alpha(g^2))=\alpha(g^1)\tens g^2,\quad g^1\gamma(g^2)=\gamma(g^1)g^2.\]
\end{proposition} 
\begin{proof}
We let $\sigma$ be the braiding for $\hat\nabla$ and set $\nabla=\hat\nabla+\beta$ for  $\beta(\omega)=-\Xi(\xi,\omega)$, which is a bimodule map by the centrality properties so that  $\nabla$ is a bimodule connection with the same $\sigma$. One can check that it is $*$-preserving by the assumed properties of $\Xi$ and $(\xi,\ )=(\ ,\xi)$, and also that $\nabla\xi$ is central. This choice of $\beta$ then solves (\ref{g1}). Moreover, given $\beta$,  the stated assumption on $\Xi$ is equivalent to the desired (\ref{g9}), so this is solved.  

Next, we note that 
\[(\nabla\xi)^1((\nabla\xi)^2,\xi)+ N\Xi^1(\Xi^2,\xi)=(\hat\nabla\xi)^1((\hat\nabla\xi)^2,\xi)-n\Xi^1(\Xi^2,\xi)={1\over 2}\extd(\xi,\xi)-n\Xi^1(\Xi^2,\xi)\]
where we used for the first equality the definition of $\beta$ and for the second equality we used (\ref{cenxi}) to remove the braiding in
\[ \extd(\xi,\xi)=(\hat\nabla\xi)^1((\hat\nabla\xi)^2,\xi)+(\id\tens (\xi, ))(\sigma(\xi\tens (\hat\nabla\xi)^1)\tens(\hat\nabla\xi)^2)= 2(\hat\nabla\xi)^1((\hat\nabla\xi)^2,\xi)\]
The first equality is an instance of the metric compatibility of $\hat\nabla$ in terms of $(\ ,\ )$ in \cite[Eq.~(8.9)]{BegMa}. 
With this in hand, we evaluate (\ref{g3})  against $( ,\xi)$ to obtain
\[ \alpha(\xi)+C(\xi,\xi)+ {1\over 2}\extd(\xi,\xi)-n\Xi^1(\Xi^2,\xi)=0.\]
We subtract half of (\ref{g5}) to get an equation for $C$, which immediately solves as 
\begin{equation}\label{genC} C=-{1\over 2n}\extd n+ \Xi^1(\Xi^2,\xi).\end{equation}
Note that $\extd(\xi,\xi)$ is central due to $\hat\nabla\xi$ central and the expression above, and we assumed $\extd N$, $\Xi$ are, hence $C$ is central. Given $C$,  (\ref{g3}) uniquely solves for $\alpha$ as $\cg$ is invertible, and the result is a bimodule map as $\nabla\xi$ is central. The other half (\ref{g2}) is then equivalent to the first stated condition. 

Similarly, given $C$, (\ref{g6}) becomes
\begin{equation}\label{g10} \dot\xi +E\xi- {\extd N\over 2} + g^1\gamma(g^2).\end{equation}
Evaluating against $(\xi, )$ and subtracting half of (\ref{g8}), we get an equation for $E$ which solves as
\begin{equation}\label{genE} E={1\over 2n}\left(\dot N - 2(\xi,\dot\xi)+ (\xi,\extd N)\right).\end{equation}
 Then (\ref{g10}) uniquely solves for $\gamma$ as $\cg$ is invertible, which then solves (\ref{g6}). The construction gives $E$ as central and from this that $\gamma$ is a bimodule map. The other half (\ref{g7}) is then equivalent to the second condition stated. 

Finally, it is clear from the formulae obtained that $C,E$ are self-adjoint. Then, as remarked above (comparing $*$ of (\ref{g3}) with (\ref{g2})) it follows that $\alpha$ commutes with $*$. Similarly for $\gamma$. \end{proof}

This lays out the solution strategy, which can also apply more loosely if some of the conditions are not met. The main issue is if we can find $\Xi$ obeying the stated condition, or $\Xi,\nabla$ such that (\ref{g9}) holds. We will see how this plays out for the specific spatial QRGs that we consider. Because  we don't know in general if we can solve for $\Xi$, it is useful to make the following definition.

\begin{definition} We distinguish three cases:

Type I:  $\xi=0$, $\nabla=\hat\nabla$. There are no `shifts'  so $\tilde \cg$ is block diagonal, and there may be are no constraints on $\dot\cg$ if we can always solve for $\Xi$.

Type II:   $\Xi=0$, $\nabla=\hat\nabla$. There are `shifts' but $\dot\cg$ is constrained by (\ref{g9}).

Type III: $\xi,\Xi\ne 0$. Then typically $\nabla$ is not a QLC for $\cg$ and there are are no constraints on $\dot\cg$ if we can solve for $\Xi$.

\end{definition}
There is a small intersection of Types I,II but this entails $\dot\cg=0$ and the solution is
\begin{equation}\label{type12} \alpha=0,\quad C=-{1\over 2N}\extd N,\quad E=-{1\over 2N}\dot N,\quad\gamma(\omega)={1\over 2}(\extd N,\omega)={1\over 2}(\omega,\extd N)\end{equation}
requiring for its existence that the two versions of $\gamma$ coincide. The general Type I case is not much harder, we just have to find $\Xi$ to match
\begin{equation}\label{type1a} \wedge(\Xi)=0,\quad \dot\cg=N(\id+\sigma^{-1})\Xi\end{equation}
then the rest of the QLC, if it exists, is prescribed by 
\begin{equation}\label{type1b} \alpha(g^1)\tens g^2=\sigma(g^1\tens \alpha(g^2))= -N\Xi\end{equation}
and $\gamma, C,E$ as in (\ref{type12}).

\section{Unique QLC in the central basis case} \label{seccen}

In this section, we look at a class of spatial QRGs where everything can be solved following the template of Propsition~\ref{prop:gen} and we have a spacetime QRG including a quantum Levi-Civita connection for every spacetime quantum metric $\cg$ of the ADM form. We start with the strongest assumptions and hence the easiest case, namely we assume:

(i) $A$ has trivial centre (to start with, later dropped);

(ii) $\Omega^1=A.\{s^i\}$ i.e. has a basis $s^i$ and these are self-adjoint and central;

(iii) The $s^i$ anticommute in $\Omega$ (which is then $A$ times their Grassmann algebra).

In this case $\cg=g_{ij}s^i\tens s^j$ needs $g_{ij}$ central by (ii) and hence constants in $A$ by (i) (they can still depend on $t$). Moreover, the reality condition becomes clearly that $\overline{g_{ij}}=g_{ji}$, i.e. a hermitian matrix. We also suppose quantum symmetry $\wedge(g)=0$ which given (iii) means $g_{ji}=g_{ij}$, i.e. a real matrix. It should also be non-degenerate and in most case we would assume positive (though this is not needed here).  We could just impose symmetry of the $g_{ij}$ directly in place of (iii), as this is what we really assume. The inverse metric is given by the  matrix inverse $g^{ij}=(s^i,s^j)$.

In this context, it is natural to impose symmetry of $\tilde g$ under $\wedge$ (or directly) as $\eta=\xi$, and we do so henceforth. Centrality of $\tilde g$ needs $\xi=\xi_is^i$ central, hence $\xi_i$ constant in $A$ by (i), and ditto for $N$. Note that in the symmetric case,  (\ref{g6})-(\ref{g7}) are the same as $\gamma(g^1)\tens g^2=\gamma^i g_{ij}= g_{ij}\gamma^j=g^1\gamma(g^2)$ where $\gamma=\gamma_is^i$ must have $\gamma_i$ constants in $A$ by (i) and that $\gamma$ is a bimodule map in the assumptions of Proposition~\ref{prop:sig}. Similarly from there and (i), we have $\alpha(s^i)=\alpha^i{}_js^j$, $C=C_is^i$ and $\Xi=\Xi_{ij}s^i\tens s^j$ with coefficients, as well as $E$ and the coefficients of $\sigma(s^i\tens s^j)=\sigma^{ij}{}_{kl}s^k\tens s^l$, all for now constant in $A$. Hence, under the assumptions (i)-(iii), everything reduces to a system of equations for matrices, vectors and functions of $t$ only!  Moreover, for simplicity we focus on the case where

(iv)  $\sigma(s^i\tens s^j)=s^j\tens s^i$, i.e. the flip on the basis 1-forms.

This is not the only case of interest but suffices for the fuzzy sphere example later. We also recall that  
\begin{equation}\label{n} n:= (\xi,\xi)-N.\end{equation}
is the actual (squared) lapse function and we have repurposed $N$ for the coefficient of $\extd t\tens\extd t$ in $\tilde g$. We assume that $N, n$ are always invertible during the evolution.

\begin{theorem}\label{thm:cen} For $(A,\cg,\hat\nabla)$ a QRG with the assumptions (i)-(iv) and if $N$, $n$ are invertible then the quantum metric (\ref{tg}) has a unique QLC $\tilde\nabla$.
\end{theorem}
\begin{proof} We solve for the QLC in stages, finding it explicitly in terms of the initial data and a free choice of $\cg,N,\xi$ subject to our invertibility and `reality' (self-adjointness) conditions. The first step is the critical one after which we could appeal to Proposition~\ref{prop:gen}, but we proceed in a more  self-contained manner to arrive at more explicit formulae for our case. 

(1) Under our assumptions, the condition $\wedge(\Xi)=0$ for zero torsion combined with (iii) says that $\Xi_{ij}$ are symmetric in their indices. Then, given (iv), $\sigma$ acts as the identity, hence we can immediately solve (\ref{g1}) by
\[ \Xi_{ij}={1\over 2N} \left(\dot g_{ij}+\xi_k(\Gamma^k{}_{ij}+\Gamma^k{}_{ji})\right),\]
where $\nabla s^i=-\Gamma^i{}_{jk}s^j\tens s^k$. Although we do not know that $\Gamma^i{}_{jk}$ are necessarily central (hence constants if we assume (i)), we know that $\nabla\xi$ is in the centre as $\xi$ is. This is because  $[a,\nabla\xi]=\sigma([a,\xi])+\sigma(\xi\tens\extd a)-\extd a\tens\xi=0$ since $\xi$ is in the centre and by our assumption (iv). Hence $\xi_k\Gamma^k{}_{ij}$ are in the centre, so this defines $\Xi_{ij}$ also in the centre. 

Note that if $\nabla$ is a QLC then we must be in Type I or Type II. For in this case, if $\xi=0$ then $(\ref{g1})$  is empty. Otherwise, (\ref{g1}) can be written as $\Xi_{ij}\xi+\Xi_{ik}\xi_j s^k=0$ for all $i,j$. Let $\xi_j$ be a nonzero component, then we see that $\Xi_{ik}s^k$ is a 1-form proportional to $\xi$. Hence $\Xi=\zeta\tens\xi$ for some 1-form $\zeta$. Then (\ref{g1}) become $\zeta\tens (\xi\tens\xi)=0$, so $\Xi=0$. 

(2) Hence to proceed more generally, we let 
\[ \nabla=\hat\nabla + \beta;\quad \beta:\Omega^1\to \Omega^1\tens_A\Omega^1,\]
where $\hat\nabla$ is a QLC for $\cg$ and suppose that $\beta$ is a bimodule map. This is required by the connection properties if we want the same $\sigma$ and will be sufficient for the solution. We then set
\[ \beta(g^1)\tens g^2=-\Xi\tens\xi,\quad  g^1\tens\beta(g^2)=-\xi\tens\Xi,\]
which are equivalent (given that $\cg$ is invertible) to 
\begin{equation}\label{betacen} \beta(\omega)=-(\omega,\xi)\Xi=-\Xi(\xi,\omega)\end{equation}
for all $\omega\in \Omega^1$. The two halves are equivalent because $\Xi$ is central and $(\omega,\xi)=(\xi,\omega)$ because $\xi$ is central and $g^{ij}$ is symmetric. This then solves (\ref{g1}) and hence allows us to write $\Xi_{ij}$ in terms of $\beta$ and the corresponding $\hat\Gamma^i{}_{jk}$ for $\hat\nabla$. Doing so then gives
\begin{equation}\label{Xicen} \Xi_{ij}=-{1\over 2n}\left(\dot g_{ij}+\xi_k(\hat\Gamma^k{}_{ij}+\hat\Gamma^k{}_{ji})\right)\end{equation}
provided $n\ne 0$ as we assume.

(3) Next, observe using that $g_{ij}$ is symmetric for the second equality,
\[ \alpha(g^1)\tens g^2=\alpha^j{}_i g_{jk} s^i\tens s^k=\sigma(g_{kj}\alpha^j{}_i s^k\tens  s^i)=\sigma(g^1\tens \alpha(g^2)).\]
Hence, given our $\tilde g$ symmetry and $\sigma$ assumptions, (\ref{g2}) and (\ref{g3}) are the same so we just need to solve one of them. We also need $N$ to be invertible for $\tilde g$ to be, and it is constant in spatial directions due to (i) so $\extd N=0$. Hence (\ref{g5}) is solved by
\[ C= -N^{-1} \alpha(\xi)\]
and (\ref{g3}) by defining $\alpha$  according to
\[ \alpha(g^1)\tens g^2=-\nabla\xi-C\tens\xi-N\Xi\]
which requires 
\[ \alpha(\xi)=-\hat\nabla(\xi)^1(\hat\nabla(\xi)^2,\xi)-\beta(\xi)^1(\beta(\xi)^2,\xi)- C(\xi,\xi) - N\Xi^1(\Xi^2,\xi)\]
where $\Xi=\Xi^1\tens\Xi^2$ and ditto for $\hat\nabla(\xi),\beta(\xi)$ is notation. The $\beta$ term and the last term combine and we as in the proof of Proposition~\ref{prop:gen} then use metric compatibility. In our present case it is clear that the braiding $\sigma$ of $\xi$ with anything is the flip and likewise $(\xi,\ )=(\ ,\xi)$, given that $\xi$ is central. Then  $\extd(\xi,\xi)= \hat\nabla^1\xi(\hat\nabla^2\xi,\xi)+ \hat\nabla^1\xi(\xi,\hat\nabla^2\xi)=2 \hat\nabla^1\xi(\hat\nabla^2\xi,\xi)$ in our case (even if there is a centre). But $(\xi,\xi)$ is a constant under our assumption (i) so $\extd(\xi,\xi)=0$ and we can drop the first term in $\alpha(\xi)$. This can also be seen directly where, by the same arguments as in \cite{LirMa}, one has that metric compatibility under our assumptions amounts to 
\begin{equation}\label{gGamma} \hat\Gamma_{ijk}+\hat\Gamma_{kji}=0\end{equation}
with indices lowered/raised by $g_{ij}$ or its inverse $g^{ij}$.  Using our formula for $\beta(\xi)$ and inverting $C$ in terms of $\alpha(\xi)$ to eliminate the latter,
we arrive at 
\begin{equation}\label{Ccen}  C=\Xi^1(\Xi^2,\xi),\quad \alpha(g^1)\tens g^2=-\hat\nabla\xi + n\Xi - \Xi^1(\Xi^2,\xi)\tens\xi\end{equation}
again provided $n\ne 0$. This is in line with (\ref{genC}).

(4) The remaining 2 equations (since (\ref{g6}),(\ref{g7}) are the same under our symmetry assumptions) as solved by first evaluating (\ref{g7}) against $(\ ,\xi)$ to find $\gamma(\xi)$ and then solving for $E$, giving 
\[ \gamma(\xi)=-(\dot\xi+E\xi,\xi),\quad E={\dot N-2(\dot\xi,\xi)\over 2n},\]
and $E$ then gives $\gamma(g^1)g^2$ from (\ref{g7}). This is in line with (\ref{genE}). Next observe that
\[ {\extd (\xi,\xi)\over\extd t}=(\dot\xi,\xi)+ (\xi,\dot\xi)- \xi_i\xi_j g^{ik}g^{lj}\dot g_{kl}=2(\dot\xi,\xi)\]
because the $\dot g_{kl}$ term, on substituting in terms of $\hat\Gamma$ due to the stated evolution for the metric, vanishes by the antisymmetry in (\ref{gGamma}). In this way, we obtain
\begin{equation}\label{Ecen}
E=-{\dot n\over 2 n}+(\Xi,\xi\tens\xi),\quad \gamma(\omega)=-(\dot\xi+E\xi,\omega).
\end{equation}

One can check that this solution also meets all the conditions with respect to $*$ in Corollary~\ref{cor:star}. Here, $s^i$ self-adjoint and $\sigma$ the flip in basic forms  requires $\Xi_{ij}$ to be real as they are, as well as $N$, $\xi_i$, $\alpha^i{}_j$ real as they are. Here, $g_{ij}$ are real as part of the spatial QRG respecting $*$ and likewise for the $\hat\Gamma^i{}_{jk}$ under assumption (iv). 
 \end{proof}
 
Putting everything in matrix terms, we can summarise our solution contained in the above constructive proof as
  \begin{equation}\label{nablacen1} \nabla s^i=\hat\nabla s^i-\xi^i \Xi_{jk}s^j\tens s^k+ \alpha(s^i)\tens_S\extd t -g^{ij}(\dot\xi_j-{\dot n\over 2n}\xi^j)\extd t\tens\extd t,\end{equation}
 \begin{equation}\label{nablacen2} \nabla \extd t=\Xi_{ij}s^i\tens s^j+\Xi_{ij}\xi^j\tens_S\extd t+ \left(\xi^i\xi^j\Xi_{ij}- {\dot n\over 2n}\right)\extd t\tens\extd t, \end{equation}
 where we raise or lower with $g^{ij},g_{ij}$, $\Xi_{ij}$ is in (\ref{Xicen}) and
 \begin{equation}\label{alphascen} \alpha(s^i)=\left( \xi_l\hat\Gamma^l{}_{kj}g^{ij} + \Xi_{kj}(n g^{ij}- \xi^i\xi^j)\right)s^k.\end{equation}

As mentioned in the introduction, we have adjoined a single commutative variable $t$ but this can be replaced by a commutative spatial variable or even a classical manifold coordinate algebra, which is then the Kaluza-Klein analysis for a product (or tensor product of coordinate algebras)  as achieved for the fuzzy sphere in \cite{LiuMa}.  

\subsection{Special cases and application to the fuzzy sphere}

The general solution has two interesting special cases where  $\nabla=\hat\nabla$ is an chosen QLC for $\cg$, ie the where the spacetime QLC restricts to the
spatial one. 

\begin{corollary}\label{cor:diag} (Type I special case)  If $\xi=0$ (i.e.,  if we do not seek to introduce `shift' cross terms) then we can choose $\cg(t)$ and $N(t)$ freely with spacetime QLC
\[\tilde\nabla s^i=\nabla s^i -{1\over 2} g^{ij}\dot g_{jk}s^k \tens_S\extd t,\quad \tilde\nabla\extd t={1\over 2N}(\dot g -\dot N \extd t\tens\extd t).\]
\end{corollary}
\begin{proof} The general solution immediately reduces to 
\[ C=\gamma=0,\quad E=-{\dot N\over 2 N},\quad \Xi={1\over 2} N^{-1}\dot\cg,\]
\[ \alpha(\omega)=\alpha(g^1)(g^2,\omega)=-N\Xi^1(\Xi^2,\omega)=-{1\over 2}\dot g^1(\dot g^2,\omega)\]
which translates into the $\tilde\nabla$ stated in our case since $(s^i,s^j)=g^{ij}$ is the inverse matrix to $g_{ij}$ so 
$\alpha(s^i)= -{1\over 2}\dot g_{jk} s^j(s^k,s^i)= -{1\over 2}\dot g_{jk}g^{ki}s^j$.   \end{proof}

For a concrete example, we take the unit fuzzy sphere. This has coordinate algebra 
\[   [y^i,y^j]=2\imath\lambda_p \eps_{ijk}y^k,\quad \sum_i (y^i)^2=1-\lambda_p^2 \]
with a central basis of three 1-forms $\{s^i\}$ for the rotationally invariant calculus given in \cite[Ex.~1.46]{BegMa}. The parameter $\lambda_p$ is dimensionless but if it has value $1/n$ then additional relations  can be add that render the above into an algebra isomorphic to $M_n(\C)$ via the standard $n$-dimensional matrix representation of $U(su_2)$. In both cases, the algebra has trivial centre when $\lambda_p\ne 0$ and the $\{s^i\}$ anticommute in the exterior algebra. One also has $\sigma(s^i\tens s^j)=s^j\tens s^i$ for the QLC. Hence all the conditions in the preceding two sections apply and we can immediately write down the spacetime QRG from Corollary~\ref{cor:diag} in the easy case where $\xi=0$ (so that the spacetime quantum metric is block-diagonal) for any $N(t),\cg(t)$ chosen freely (but invertible). A special case $N=-1$ recovers \cite[Thm.~3.1]{ArgMa}, such as the 3D FLRW model with $g_{ij}(t)=R^2(t)\delta_{ij}$ for an expanding fuzzy sphere.

\begin{corollary}\label{cor:evg}  (Type II special case) If $\Xi=0$ then $\xi(t),N(t)$ force a specific evolution for the spatial metric, 
\[  \dot g_{ij}=  -\xi_k(\Gamma^k{}_{ij}+ \Gamma^k{}_{ji}).\]
The spacetime QLC in this case is
\[ \tilde\nabla s^i=\nabla s^i+  g^{ij}\xi_l\Gamma^l_{kj}s^k\tens_S\extd t- g^{ij}(\dot\xi_j- {\dot n\over 2 n}\xi_j)\extd t\tens\extd t,\quad \tilde\nabla \extd t= -{\dot n\over 2 n}\extd t\tens\extd t.\]
\end{corollary}
\begin{proof} The general solution immediately reduces to the stated evolution from (\ref{Xicen}), $C=0$, $\alpha(g^1)\tens g^2=-\nabla\xi$, with (\ref{Ecen}) unchanged. \end{proof}

To apply this, we need the spatial QLC and there is a unique one on the fuzzy sphere with constant coefficients, given by\cite{LirMa}
\[ \Gamma^i{}_{jk}=g^{il}\left(\eps_{lkm}g_{mj}+ {{\rm Tr}(g)\over 2}\eps_{ljk}\right).\]
Hence we can choose $\xi(t), N(t)$ as we like (up to the invertibility requirements) and also $\cg(0)$ how we like, and solve the stated evolution,  which is now 
\[ \dot g_{ij}= -\xi_lg^{lk}(\eps_{kim}g_{mj}+ \eps_{kjm}g_{im})= - L_{\vec N(t)}(g)_{ij}\]
where the trace term drops out and we let $N^k:=\xi_l g^{lk}$ be the corresponding shift vector. We recognise $N^k$ times the expression in brackets as an  infinitesimal rotation $L_{\vec N}$ at time $t$ about axis $\vec N(t)$, acting on the 2-index tensor $g_{ij}$. This means that ${\rm Tr}(\cg)$ is constant under the evolution while the traceless part rotates under the action of $\vec N$ in the 5-dimensional spin 2 representation. This also means, for example, that we cannot assume that $\cg$ is diagonal because even if it starts that way, it will not remain diagonal under the evolution. Although by no means necessary for a spacetime QLC, this evolution is necessary if we want it to restrict to the spatial QLC and want to have nonzero $\xi(t)$, i.e., the Type II case. One has a similar phenomenon in classical geometry.

\subsection{Extension to nontrivial centre}\label{secZ}

Although we focussed so far on the case of trivial centre, most of the analysis in Theorem~\ref{thm:cen} works the same with nontrivial centre $Z(A)$, with $g_{ij},\xi_i, N,n,C_i,\Xi_{ij},E$ now lying here. The main difference is in the side-note that if $\nabla=\hat\nabla$ then we must by in Type I or Type II, which is no longer true as the product of two nonzero elements of $Z(A)$ could  be zero (e.g. for functions on a manifold, they could just have disjoint support). Here we just outline the extension of the proof of the theorem to this more general case using the same numbering of the steps, and list the results as a natural solution for the spacetime QLC. 

(1)/(2) $\Xi_{ij}$ is a still symmetric under assumptions (ii)-(iii) as $s^is^j$ with $i<j$ say are a basis of $\Omega^2$. Then argument for $\Xi_{ij}$ is unaffected and still makes sense as noted in the proof but now $\nabla\xi$ has derivative terms as $\xi_i$ need not be constant. The definition of $\beta$ still makes sense unchanged and we arrive now at
\begin{equation}\label{XiZ} \Xi_{ij}=-{1\over 2n}\left(\dot g_{ij}-(\del_i\xi_j+\del_j\xi_i-\xi_k(\hat\Gamma^k{}_{ij}+\hat\Gamma^k{}_{ji}))\right),\end{equation}
where the partials are defined by $\extd a=(\del_i a)s^i$  for any $a\in A$. 

(3) Get modified to 
\[ C=- N^{-1}({\extd N\over 2}+\alpha(\xi)),\quad\alpha(\xi)=-{1\over 2}\extd(\xi,\xi)+n\Xi^1(\Xi^2,\xi)\]
as $N, (\xi,\xi)$ need not be constant. Combining their differentials, we arrive at
\begin{equation} \label{CZ} C=-{1\over 2n}\extd n+ \Xi^1(\Xi^2,\xi),\quad  \alpha(g^1)\tens g^2=-\hat\nabla\xi+  {1\over 2 n}\extd n\tens\xi+ n\Xi-\Xi^1 (\Xi^2,\xi)\tens\xi.\end{equation}

(4) Gets modifed due to the change of $C$ to
\[  E={\dot N-2(\dot\xi,\xi)+(\extd N,\xi)\over 2 n},\quad \]
in the first instance and a similar adjustment to $\gamma(g^1)g^2$ in (\ref{g7}). Now when we do the same calculation we have
\[ {\extd (\xi,\xi)\over\extd t}=2(\dot\xi,\xi)-\xi^i\xi^j(-2 n\Xi_{ij}+2\del_i\xi_j-\xi^k \del_i g_{kj})= 2(\dot\xi,\xi)-2n(\Xi,\xi\tens\xi)- (\extd(\xi,\xi),\xi)\]
which now gives us
\begin{equation}\label{EZ} E=-{\dot n+ (\extd n,\xi)\over 2n}+ (\Xi,\xi\tens\xi),\quad \gamma(\omega)=- (\dot\xi+E\xi- {1\over 2}\extd N,\omega)\end{equation}
for all $\omega\in \Omega^1$. This completes the general result for the QLC without assuming trivial centre. These can be put into (\ref{tcon1})-(\ref{tcon2}) and computed in matrix terms as we did after Theorem~\ref{thm:cen}. In $\Xi_{ij}$, we now have covariant derivatives of $\xi$. In $\alpha, C$ we now have an extra $\extd n$ term, in $E$ we have an extra $\extd n$ terms and in $\gamma$ we have this from $E$ and an additional $\extd N$ term.

\begin{example} For an example of a spatial QRG with nontrivial centre, we can first use the same general construction as in Theorem~\ref{thm:cen} (or more generally as above) but with a spatial classical variable, say $r$, in place of $t$. For example, we can obtain a spatial quantum Riemannian geometry
\[ \cg= r^2 g_{ij} s^i\tens s^j+H(r)\extd r\tens\extd r,\]
\[ \nabla s^i=\nabla_{S^2} s^i - {\extd r\over r} \tens_s s^i,\quad \nabla\extd r={r\over H(r)}g_{ij} s^i\tens s^j-{1\over 2H}H'(r)\extd r\tens \extd r\]
by using Corollary~\ref{cor:diag} with $N(t)$ replaced by $H(r)$ and $\cg(r)=r^2 g_{ij}s^i\tens s^j$ as the $r$-dependent spatial metric given by scaling the initial $\cg_{S^2}=g_{ij}s^i\tens s^j$ of the unit fuzzy sphere (where $g_{ij}$ is now not being used for the coefficient matrix of $\cg(r)$).  Application of the more general result in this section to {\em this} spatial geometry -- which has as centre  the algebra of functions in $r$ -- to now adjoin $t$ with $N(r)=-\beta(r)$ then exactly recovers \cite[Thm~4.1]{ArgMa} as used there to obtain a fuzzy black hole for the relevant  $H(r),\beta(r)$. \end{example}

On the other hand, that static case for the fuzzy black-hole is a very special case as $N$  and $\cg$ in the second iteration do not even depend on $t$. Our new results, even in the diagonal Type I case, are much more powerful and allow us to construct many more quantum spacetimes complete with QLCs, for example with time-dependent $g_{ij}$ on the unit fuzzy sphere, and time dependent $H(r,t), N(r,t)$ in models like the above. Moreover, we are not limited to the block-diagonal form and and can introduce off diagonals with nontrivial $\xi$. Note that  a claim at the start of \cite[Sec.~3]{ArgMa} that off-diagonals are not allowed was an error, simply the diagonal case was a lot easier but, as we have now shown, it is not forbidden. 

\section{Some solutions for a discrete spatial QRG}\label{secdisc}

We now return to the general setting of Section~\ref{secmain} and see how the general approach to solving for a QLC there works when applied to a complementary class of spatial QRGs compared to the ones of Section~\ref{seccen}, namely discrete geometry where $A$ is functions on a graph. Here $\Omega^1$ does not have any central 1-forms at all, since a vector space basis of it is labelled by the graph arrows and these are each noncommutative (since the product by a function from the  left or right is the value at the source or target vertex of the arrow). A metric is just a real length on every edge (in the edge-symmetric case that the length does not depend on the arrow direction), see \cite[Chap.~1]{BegMa}. This is a different kind of symmetry than $\wedge(\cg)=0$ but seems better suited when one tries to solve for a QLC, albeit may need further modification when the graph has boundary legs, for example. We will limit attention to $\Z$ or $\Z_n$ where the QLC is known for any edge-symmetric metric\cite{Ma:haw,ArgMa0,BegMa}. We follow the template of Proposition~\ref{prop:gen} but in a self-contained way and to arrive at more explicit formulae. 

\subsection{Construction of a spatial QRG with $\Z_n$ for the angular geometry}\label{secZn}

In this section, we first take $\Z_n$ as a `discrete circle' and extend this by time $t$ as a Type I solution. There can be no  Type II solutions here as the QRG for $\Z_n$ does mot have any central 1-forms.  We use our construction to both recover (as a check) and extend the previous $\Z_n$-black-hole and FLRW models  \cite{ArgMa0,ArgMa}.  We then reinterpret this as an extension by $r$ to give a spatial QRG for the remaining sections. 

The QRG of the `discrete circle' $\Z_n$ for $n\ne 4$ is in fact the same as that of the lattice line $\Z$ but just has vertex labels $i$ taken modulo $n$, so what follows actually applies to either\cite{Ma:haw,ArgMa0}. These are both Cayley graphs meaning that, using the group structure, there is a natural (but not central) basis of 1-forms. In our case there are two, $e^\pm$  consisting in principle of the sum all increasing (resp. decreasing) arrows in the case of $\Z$ and clockwise resp. anticlockwise in the case of $\Z_n$. In practice, we just define $\Omega^1=A.\{e^\pm\}$ with bimodule relations and $\extd$,
\[ e^\pm f=R_\pm(f)e^\pm,\quad \extd f= \del_\pm(f) e^\pm,\quad  R_\pm(f)(i)=f(i\pm 1),\quad \del_\pm=R_\pm-\id.\]
for any function $f\in A$. The exterior algebra is the Grassmann algebra on $e^\pm$, $\extd e^\pm=0$ and the $*$-structure is pointwise complex conjugation of functions and $(e^\pm)^*=-e^{\mp}$.  The metric edge values are viewed as a real-valued non-vanishing coefficient function $g$ of the source vertices whereby
\[ \cg=g e^+\tens e^-+ R_-(g) e^-\tens e^+.\]
The shift $R_-(g)$ is needed for edge-symmetry. Then there is a unique QLC\cite{Ma:haw,ArgMa0},
\[ \nabla e^\pm=   (1-\rho_\pm)e^\pm\tens e^\pm,\quad \sigma(e^\pm\tens e^\pm)=\rho_\pm e^\pm\tens e^\pm,\quad \sigma(e^\pm\tens e^\mp)=e^\mp\tens e^\pm,\]
where $\rho_+=R_+(g)/g$, $\rho_-=R_-(R_-(g)/g))$ in a compact notation. A recent application as well as some results about bimodule connections on Cayley graphs generally are in \cite{BegMa:gra}. 

\subsubsection{Solutions for $\R\times\Z_n$ as spacetime} First, we extend the above  $\Z$ or $\Z_n$ QRG  by $t$ as a Type I solution. Thus, we let $\cg(t)$ with coefficient function $g=g(t)$ and desired `lapse' $N(t)$ depending on time and set  
\[ \tilde \cg=\cg(t)+N(t)\extd t\tens \extd t.\]
The general approach in Section~\ref{secmain} still constructs a QLC but with constraints on the time-dependence of $\cg$ even in this Type I case.  To do this, we look for the allowed $\Xi$. The kernel of $\wedge$ is spanned by $e^\pm\tens e^\pm$ and $e^+\tens_S e^-$ and as $\Xi$ needs to be central we can only have
\[ \Xi= \chi(t) e^+\tens_S e^-\]
for some real-valued time dependent $\chi(t)\in A$ such that
\[ \dot\cg=2 N(t)\chi(t) e^+\tens_S e^-.\]
Hence, evolution will break edge-symmetry unless $N\chi$ is spatially constant.  Put another way, the class of metric evolutions $\cg(t)$ which we can accommodate are of the form
\[ \cg(t)=\cg(0) + G(t)\cg_{Euc},\quad {i.e.,}\quad g(t,i)=g(0,i)+G(t);\quad  G\in C^\infty(\R),\]
\[ \cg_{Euc}=e^+\tens_S e^-;\quad \chi(t)={\dot G(t)\over 2 N(t)}.\]
Thus, for this spatial QRG, we cannot accommodate independently evolving metrics at each vertex $i$. Rather, whatever function of the vertices $g(0,i)$ is, the graph of that against $i$ just moves up or down uniformly in $i$ (but freely in time) by an amount $G(t)$. We then set $\chi$ to solve (\ref{g9}) subject to $\wedge(\Xi)=0$. 

Next solving for $\alpha$ in (\ref{g3}), we are led uniquely to 
\[ \alpha(e^\pm)=-{\dot G(t)\over 2 g_\pm(t)} e^\pm;\quad g_+:=g(t),\quad g_-:=R_-(g(t)).\]
Finally, since $C$ has to be a central 1-form, it can only be zero. Hence, (\ref{g5}) is solved by $\extd N=0$, so in fact $N(t)$ is also forced to be constant in $i$. Thus,  the rest of the QLC is given by
\[ \gamma=C=0,\quad E=-{\dot N\over 2 N}.\]
Equivalently,
\[ \tilde\nabla e^\pm=\nabla e^\pm- {\dot G(t)\over 2 g_\pm(t)}e^\pm\tens_S\extd t,\quad\tilde\nabla\extd t= {\dot G(t)\over 2 N(t)}e^+\tens_S e^- -{\dot N\over 2N}\extd t\tens \extd t.\]
A special case up to some different conventions is the $\R\times \Z_n$ FLRW model in \cite[Thm.4.1]{ArgMa0} where $i$ is mod $n$ and $N=1$ was also constant in time so there was no $\dot N$ term in the connection. Precisely the above $g(0,i)+G(t)$ structure of the allowed metric coefficient functions was an unusual constraint, and we see now how it arises as an example of a more general construction for any $N(t)$. 

\subsubsection{The same solution as a spatial QRG by adjoining $r$} \label{spatialZn}

We now take exactly the Type I solution above but with $t$ replaced by $r$ for a spatial QRG in polar coordinates, with $\Z_n$ for the angular geometry. We now use a funciton $H(r)$ in place of $N(t)$ above and $\delta(r)$ in place of $G(t)$ above. We will drop $\tilde{\ }$ since this is going to be a spatial QRG for later extension to spacetime and accordingly write $\nabla^{\Z_n}e^\pm$ for the initial one that we used before on $\Z_n$.  Then, our result above immediately gives a QRG on $\C(\Z_n)$ with (graded)commuting $r,\extd r$ adjoined as 
\[ \cg= g_+ e^+\tens e^-+ g_- e^-\tens e^++ H(r)\extd r\tens\extd r;\quad  g_\pm(r,i)= g_\pm(0,i)+\delta(r)\]
\[  \nabla e^\pm=\nabla^{\Z_n} e^\pm- {\delta'(r)\over 2 g_\pm(r)}e^\pm\tens_S\extd r,\quad\nabla\extd r= {\delta'(r)\over 2 H(r)} e^+\tens_S e^--{H'(r)\over 2H(r)}\extd r\tens \extd r.\]
As a check, this recovers exactly the spatial sector of the general analysis for static metrics with $\Z_n$ angular geometry in \cite[Thm.~5.1]{ArgMa} (i.e., ignoring the $t,\extd t$ terms there for now). The $\alpha_\pm$ in that work correspond to our $g_\pm$ and $\del_r\alpha_\pm$ to our $\delta'(r)$. We also correct  a typo in \cite[Thm.~5.1]{ArgMa} in the statement of  $\nabla e^\pm$ there (but the actual result obtained in the proof (the coefficients $\gamma_\pm$) exactly match the answer obtained now). This serves as a check on our new approach and also provides the spatial QRG for the next sections. 

\subsection{Type I spacetimes with the above spatial QRG} \label{typeIZn}

We now take the spatial QRG from Section~\ref{spatialZn} and extend {\em it} by a further $t$ using the general scheme of Section~\ref{secmain}. This time, there is one central basis vector in the spatial geometry namely $\extd r$ so we are not forced to $C=0$; we can take $N=N(t,r)$ and also have more options of $\Xi$. There appear to be plenty of solutions but as an easy special case, we can take $N=-\beta(r)$  and $\cg$ static, i.e both time independent so that we are constructing a static metric. Then we have an immediate solution of (\ref{type12})-(\ref{type1b}) etc., namely
\[ \Xi=\alpha=E=0,\quad \gamma(g^1)g^2={\beta'(r)\over 2}\extd r,\quad C=-{\beta'(r)\over 2\beta(r)}\extd r\]
which imply 
\[ \tilde\nabla\extd r=\nabla \extd r - {\beta'(r)\over 2H(r)}\extd t\tens \extd t,\quad \tilde\nabla\extd t= -{\beta'(r)\over 2\beta(r)}\extd r\tens_S\extd t,\]
where we use $\gamma(\extd r)=\gamma(g^1)(g^2,\extd r)$ using $(\ ,\ )$ for the spatial QRG, which brings in $1/H(r)$. Comparing, we see that this easy solution recovers  gives exactly the full extent of \cite[Thm.~5.1]{ArgMa} (correcting the typo) as well as the full $\Z_n$-black-hole in that work as a special case. 
On the other hand, we are not restricted to static metrics as we show with an  example.  

\begin{example} \label{FLRW} ($\Z_n$-FLRW model.) We use the spatial QRG construction as above with $H(r)=a^2/(1-k r^2)$,  $\delta(r)=-a^2 r^2$ and $g_\pm(0,i)=0$, giving
\[ \cg=a^2({1\over 1-k r^2}\extd r \tens \extd r-  r^2 e^+\tens_S e^-),\]
\[\nabla e^\pm=-{1\over r}\extd r\tens_S e^\pm,\quad \nabla\extd r=-{k r\over 1-kr^2}\extd r\tens\extd r-{r(1-k r^2)}e^+\tens_Se^-,\]
where $a$ is a constant and does not affect the QLC. We now extend this spatial QRG by $t$ using our analysis, letting $a=a(t)$ vary with time and $N=-1$ so that
\[\tilde g= -\extd t\tens\extd t+ a^2(t)({1\over 1-k r^2}\extd r \tens \extd r-  r^2 e^+\tens_S e^-)\]
is a discrete FLRW model with the sphere at each $t,r$ replaced by a flat discrete circle. To find the spacetime QLC, we set 
\[ \Xi=-{\dot a\over a }\cg,\quad \alpha={\dot a\over a}\id,\quad\gamma=C=E=0\]
to solve (\ref{type1a})-(\ref{type1b}) etc.  The resulting spacetime QLC is 
\[ \tilde\nabla\extd r= \nabla\extd r+ {\dot a\over a}\extd r\tens_S\extd t,\quad \tilde\nabla e^\pm=\nabla e^\pm+{\dot a\over a}e^\pm\tens_S \extd t,\quad  \tilde\nabla\extd t=- {\dot a\over a }\cg.\]
\end{example}
This completes the `square' of models where we have a black-hole or an FLRW $r,t$ sector and for the angular sector a fuzzy sphere or $\Z_n$, the other three being in \cite{ArgMa}. The connection coefficients in all four are not too different from their classical values possibly with a `dimension jump', because although they are
noncommutative in different ways, the rotationally symmetric metrics considered 
mean that this aspect does not show up in the QRG itself. It would, however, show up for the behaviour of fields that are not rotationally invariant moving on the
quantum spacetime.

\subsection{Type II spacetimes with the above spatial QRG} \label{typeIIZn}

We now continue with the spatial QRG of Section~\ref{spatialZn} but look for Type II solutions where we adjoin $t,\extd t$ with $\Xi=0$. To keeps things simple, we will assume that
the metric is always constant going around the $\Z_n$, so $g_\pm(0,i)=0$ and $g_\pm=\delta(t,r)$. We also assume that the rest of the geometry is also rotationally invariant, so  $N=N(t,r)$ and $\xi= \xi(t,r)\extd r$ where there is only one central 1-form $\extd r$ so only this component of $\xi$, and we are allowing everything to depend on time.  So the spacetime metric is 
\begin{equation}\label{tgZn} \tilde \cg=N(t,r)\extd t\tens\extd t+\xi(t,r)\extd r\tens_S\extd t+  H(t,r)\extd r\tens\extd r+  \delta(t,r)e^+\tens_S e^-\end{equation}
for some functions $N,\xi,H,\delta$ of two variables. 

We want to find a QLC for this and we look for one of Type 2 with spatial metric the last two terms above at each $t$ and at each $t$ the spatial QLC from Section~\ref{spatialZn} which we write in a short form as follows. Since we took the constant metric on $\Z_n$ we have $\nabla^{\Z_m}e^\pm=0$ and
\[ \nabla e^\pm=- A(t,r)e^\pm\tens_S \extd r,\quad \nabla \extd r=B(t,r) e^+\tens_Se^-- D(t,r)\extd r\tens\extd r,\]
\[ A(t,r)={\delta'(t,r)\over 2\delta(t,r)},\quad B(t,r)={\delta'(t,r)\over 2 H(t,r)},\quad D(t,r)={H'(t,r)\over 2 H(t,r)}.\]
We are also going to stop writing the $t,r$ dependencies, but these should be understood with radial derivatives given by a prime and time ones by dot, and to compensate we underline the 1-forms. Then
\[\nabla\underline{\xi}=(\xi'-\xi  D )\extd r\tens\extd r+\xi  B  e^+\tens_Se^-\]
and $\sigma$ acts as the identity on $e^+\tens_S e^-$ and on $\extd t\tens\extd t$ which dictates $\dot \cg$ to solve (\ref{g9}) according to
\[ \dot H= 2 ( \xi' -\xi D),\quad  \dot\delta=2\xi  B,\]
or explicitly
\begin{equation}\label{HevolZn} \dot H+ {H'\xi\over H}=2\xi',\quad \dot\delta-{\delta'\xi\over H}=0.  \end{equation}
We also set $\underline{C}=C \extd r$ for a rotational
invariant connection and ditto will look for $E=E(r)$  (all depending on $t$). If we take $C$ of this form then (\ref{g2})-(\ref{g3}) are the same and solved by
\[\alpha(\extd r)={(D -C )\xi  -\xi' \over H }\extd r,\quad \alpha(e^\pm)=- {\xi  B \over\delta }\]
if we know $C $. But (\ref{g5}) then tells us that 
\[ C n ={N' \over 2 }+{\xi  \over H }(D \xi  -\xi'  );\quad n ={\xi^2\over H }-N \]
which solves for $C $ provided $n\ne 0$. Putting in the form of $D,B$, the final result is
\begin{equation}\label{CsolZn} C =-{n' \over 2 n },\quad \alpha(\extd r)={1\over H}\left({n'\xi \over 2 n   }+ {H'\xi \over 2 H}-\xi'\right)\extd r,\quad \alpha(e^\pm)=-{\xi \delta' \over 2 H \delta } e^\pm. \end{equation}

Similarly, we solve the first of (\ref{g6})-(\ref{g7}) after cancelling $\alpha(\xi)$ via (\ref{g5}), by
\[ \gamma(\extd r)={1\over H }(-\dot\xi +{N' \over 2}-E \xi ),\quad \gamma(e^\pm)=0  \]
if we know $E $ but the last of (\ref{g8}) then tells us that
\[ E  n ={\dot N \over 2}+{\xi ( N' -2 \dot\xi )\over 2 H }\]
which now solves for $E $. These can be simplified to
\begin{equation}\label{EsolZn} E= -{1\over 2n}\left(\dot{n}+ {n'\xi\over H}\right),\quad  \gamma(\extd r)={1\over 2 nH }\left(2 N \dot\xi -\dot N\xi  -N N' \right),\quad \gamma(e^\pm)=0. \end{equation}

The bottom line is that we can choose $\xi ,N $ freely at all times but need $H,\delta$ to evolve according to (\ref{HevolZn}), then we have solved for the spacetime QLC.  

\subsection{General solution for $\Z_n$ rotationally invariant metrics} 

Alternatively,  if we are not interesting in this natural constraint on the  evolution in the quantum metric, we can just keep an arbitrary rotationally invariant form (\ref{tgZn}), and set 
\begin{equation}\label{XisolZn} \Xi={1\over 2 N}\left((\dot H+ {H'\xi\over  H} -2\xi')\extd r\tens \extd r+ ({\dot \delta}-{\delta'\xi \over H})e^+\tens_Se^-\right) \end{equation}
and then we can more directly solve (\ref{g3}) replacing $N\Xi+\nabla\xi=\dot\cg/2$ to get an equation for $\alpha, C$. Then (\ref{g5}) provides another such equation, with joint solution
\begin{equation}\label{CgenZn} C={1\over 2 n}\left( N'- { \dot H \xi\over H}\right),\quad\alpha(\extd r)={1\over 2 nH}\left(N\dot H- N'\xi\right),\quad \alpha(e^\pm)=-{\dot\delta\over 2\delta}.  \end{equation}
Similarly, we then solve (\ref{g6}) and (\ref{g8}) jointly to get $\gamma$ as in (\ref{EsolZn}) unchanged and 
\begin{equation}\label{EgenZn} E={1\over 2nH}\left(\dot N H- 2\xi\dot\xi+N'\xi\right). \end{equation}
One can can check that this reduces to the previous section when (\ref{HevolZn}) holds. This then provides the spacetime QLC for any rotationally invariant metric (\ref{tgZn}) with $\Z_n$ angular geometry.

\section{Conclusions}\label{seccon}

We have provided a framework for the construction of time-sliced quantum Riemannian geometries (QRG) where the spatial slices are a fixed differential exterior algebra $(A,\Omega,\extd)$ but with spatial metrics, shift 1-forms and lapse functions on $A$ that are allowed to vary in time. We limited ourselves to $t,\extd t$ classical and (graded) commuting with the spatial sector. We provided a general route to finding the spacetime quantum Levi-Civita connection (QLC) and then solved it completely in the case of a central basis of $\Omega^1$ and some assumptions in the exterior algebra and the braiding $\sigma$ of the spatial QLC. This class includes fuzzy spheres, but we also showed that the scheme gives useful solutions for discrete space QRGs where there may be no central 1-forms, focussing on the discrete circle $\Z_n$. If,  as is plausible, Planck scale noncommutativity of spacetime is induced by quantum gravity effects, this would impact Hamiltonian quantization of any quantum field theory via a time-slice approach, including quantum gravity itself. It remains to be considered how much, if any, of the ADM formalism and the Wheeler-De-Witt equation can still be achieved on a highly noncommutative spacetime and without (yet) a full understanding of variational calculus in noncommutative geometry. Notably, comparing $\dot\cg$ in  Proposition~\ref{prop:gen} or  (\ref{XiZ}) with the usual 
\[ \dot g_{ij}= {2N\over\sqrt{g}}(\pi_{ij}-g_{ij}{{\rm Tr}(\pi)\over 2})+\hat\nabla_iN_j+\hat\nabla_j N_i\]
for momentum $\pi_{ij}$ in the ADM formalism, tells us that $\Xi_{ij}$ should be identified up to a sign and scaling by a measure with this combination of $\pi_{ij}$ and its trace, and in geometric terms be closely related to the extrinsic curvature. Eventually, Hamiltonian quantisation on quantum spacetimes, including canonical quantum gravity adapted to this case, should tie up with functional integral approaches. This is a key topic for further work. 

Meanwhile, our results provide a plentiful supply of quantum spacetimes constructed in this way, complete with QLCs, on which known constructions for curvature, Laplacians, Dirac operators etc. could be computed and various ideas for physics on  quantum spacetimes explored. We recovered the fuzzy FLRW, fuzzy black-hole and $\Z_n$ black-hole from \cite{ArgMa} as easy applications of the formalism (as opposed to rather hard direct calculations as previously) and provided now a $\Z_n$-FLRW model with a circle at each $r,t$. Our work can also be seen as a special case of the Kaluza-Klein construction in \cite{LiuMa} with spacetime just $\R$, but going beyond the fuzzy sphere case there and with a different interpretation of $A$ as an internal geometry. Iterating the construction in Section~\ref{secZ} in principle to further spatial variables is then on a par with the Kaluza-Klein result.  In terms of concrete applications, within the family of quantum spacetime models that we can now construct by the results here, it should be possible to include a fuzzy or discrete version of the Oppenheimer-Snyder metric for black-hole collapse\cite{Opp} for example. This will be looked at elsewhere. 

In further work, it would be interesting to apply the time-slicing formalism to other spatial QRGs. One limitation in the $\Z_n$ case was that we were a bit limited if we wanted metrics at different places on the circle to evolve independently. This can be traced back to the
$\wedge(\Xi)=0$ condition needed for zero torsion not playing well with the fact that the natural metrics on graphs such as $\Z_n$ are edge-symmetric not $\wedge$-symmetric. It is possible to solve the spatial QRG here more generally, see \cite{Ma:haw} for a class of solutions with $\wedge$-symmetry, or \cite{Sit} for a class of QRGs for other metrics on $\Z_n$ including non edge-symmetric ones. An alternative could be to allow a degree of torsion by dropping or modifying $\wedge(\Xi)=0$ so as to allow 
general metric evolution. It should also be of  interest to take $\Z_2\times\Z_2$ with its QRG from \cite{Ma:squ}, where the QRG for a general (edge-symmetric) metric is known and these are also $\wedge$-symmetric, so this should be less restrictive. Another class of examples that could be considered is $A=M_2(\C)$ with calculus from \cite[Ex.~1.3.7]{BegMa} very different from that of a fuzzy sphere. Here, $\Omega$ is not  a Grassmann algebra and $\sigma$ for the known QLCs is not typically a flip, hence this is a rich case that goes beyond Section~\ref{seccen} but still has the trivial centre and central basis features which made that case more easily computable. 

Finally, we would like to be able to have nontrivial commutation relations between $t$ and the spatial sector. An example that has been studied is when $[a,t]=\lambda X(a)$ for a vector field $X$ and $a\in C^\infty(M)$. In this case a spacetime calculus and wave operator, but not a QRG in the modern sense, was found with $M$ the spatial geometry of a black-hole\cite{Ma:alm}, which already tells us that a general theorem for $t$ non-central case, while desirable, will be difficult if we want an actual QRG. We could also could equip $t$ with a nonstandard differential calculus or replace $C^\infty(\R)$ by another QRG entirely, such as functions on a lattice line, with or without $t$ central. Even the QLC for a general metric on $\Z\times \Z$ (where $t$ is now a commuting discrete variable) does not appear to exist, possibly for similar restrictions as in the present work in the case of discrete space.  Understanding the full extent of the obstruction here and resolving it is another important direction for further work. 

\section*{Acknowledgements} I thank Edwin Beggs for some initial discussions. The project was supported by a Leverhulme Trust project grant RPG-2024-177.

\section*{Declarations}


\medskip
\noindent{\bf Data availability:} Data sharing is not applicable as no data sets were generated or analysed during the current 
study.

\medskip
\noindent{\bf Conflict of Interest:} The author has no competing  interests to declare that are relevant to the content of this article.


\begin{thebibliography}{99}

\bibitem{ArgMa0} J. Argota-Quiroz and S. Majid, Quantum gravity on polygons and $\R \times \Z_n$ FLRW model, Class. Quantum Grav. 37 (2020) 245001 (43pp)

\bibitem{ArgMa} J. Argota-Quiroz and S. Majid, Fuzzy and discrete black hole models, Class. Quantum Grav. 38 (2021) 145020 (36pp)

\bibitem{ADM}R. Arnowitt, S. Deser and C.W. Misner, in {\em Gravitation: An Introduction to Current Research}, ed. L. Witten, (Wiley, 1962); arXiv:0405109 (gr-qc)


\bibitem{BegMa}
E. J. Beggs and S. Majid, \textit{Quantum Riemannian Geometry}, Grundlehren der mathematischen Wissenschaften, vol. 355, Springer, 2020

\bibitem{BegMa:gra} E. Beggs and S. Majid, Quantum geodesic flows on graphs, Lett. Math. Phys. (2024) 114:112 (41pp)

\bibitem{BliMa} S. Blitz and S. Majid, Quantum curvature fluctuations and the cosmological constant in a single plaquette quantum gravity model, Class. Quantum Grav. Lett. 42 (2025) 04LT01 (12pp)
 
\bibitem{Sit}A. Bochniak, A. Sitarz and P. Zalecki, Riemannian geometry of a discretized circle and torus, SIGMA 16 (2020), 143 (28pp)


\bibitem{DFR}S. Doplicher, K. Fredenhagen and J. E. Roberts, The quantum structure of spacetime at the Planck scale and quantum fields, Commun. Math. Phys. 172 (1995) 187--220

\bibitem{Will}H. Hamber and R. Williams, Discrete Wheeler-De-Witt equation, Phys. Rev. D 84 (2011) 104033

\bibitem{Hoo}G.'t Hooft, Quantization of point particles in 2+1 dimensional gravity and space-time
discreteness, Class. Quant. Grav. 13 (1996) 1023

\bibitem{LirMa}E. Lira-Torres and S. Majid, Quantum gravity and Riemannian geometry on the fuzzy sphere, Lett. Math. Phys. (2021) 111:29 (21pp)

\bibitem{LiuMa} C. Liu and S. Majid, Yang-Mills fields from fuzzy sphere quantum Kaluza-Klein model, J. High Energ. Phys. 07 (2024) 195

\bibitem{LiuMa1} C. Liu and S. Majid, Kaluza-Klein ansatz from Lorentzian quantum gravity on the fuzzy sphere, Euro. Phys. Jour. C 85  (2025) 1464 (12pp)

\bibitem{Mad}J. Madore, The fuzzy sphere. Class. Quantum Grav. 9 (1992)  69--88

\bibitem{Ma:pla} S. Majid, Hopf algebras for physics at the Planck scale, Class. Quant. Grav. 5 (1988) 1587--1607

\bibitem{Ma:alm} S. Majid, Almost commutative Riemannian geometry: wave operators, Commun. Math. Phys. 310 (2012) 569--609

\bibitem{Ma:squ} S. Majid, Quantum gravity on a square graph, Class. Quantum Grav 36 (2019) 245009 (23pp) 

\bibitem{Ma:haw} S. Majid, Quantum Riemannian geometry and particle creation on the integer line, Class. Quantum Grav 36 (2019) 135011 (22pp) 

\bibitem{Ma:are}S. Majid, Quantum gravity:are we there yet? Phil. Trans. Roy. Soc. A 383 (2025) 20230377 (20pp)

\bibitem{MaRue} S. Majid and H. Ruegg, Bicrossproduct structure of the $\kappa$-Poincar\'e group and non-
commutative geometry, Phys. Lett. B. 334 (1994) 348--354

\bibitem{MaSim}S. Majid and F. Simao, Quantum variational calculus on a lattice, arXiv: 2508.02628 (hep-th)

\bibitem {Opp}J.R. Oppenheimer and H. Snyder, On continued gravitational contraction,  Phys. Rev. 56 (1939) 455--459


\bibitem{Sny}H.S. Snyder, Quantized space-time, Phys. Rev. 71 (1947) 38--41



\end{thebibliography}
\end{document}